\documentclass{article}

\usepackage{amsmath,amsthm,amssymb,bm,color,graphicx,mathrsfs} 
\usepackage{amscd,verbatim}
\usepackage[breaklinks=true,colorlinks=true,linkcolor=blue,urlcolor=blue,citecolor=blue]{hyperref}
\usepackage{graphicx}
\usepackage{subfigure}
\usepackage{comment}
\usepackage{lipsum}
\usepackage[arrow,matrix,curve]{xy}
\usepackage{authblk}
\usepackage[numbers]{natbib}
\usepackage{braket}
 \usepackage{slashed}
\usepackage{tikz}
\usetikzlibrary{decorations.markings,calc,fadings,decorations.pathreplacing}

\newcommand\wt{\widetilde}
\newcommand\wh{\widehat}

\newcommand{\CC}{{\mathbb C}}
\newcommand{\PP}{{\mathbb P}}
\newcommand{\RR}{{\mathbb R}}
\newcommand{\TT}{{\mathbb T}}
\newcommand{\ZZ}{{\mathbb Z}}
\newcommand{\NN}{{\mathbb N}}
\newcommand{\HH}{{\mathbb H}}

\newcommand{\vect}[1]{\boldsymbol{#1}}

\newtheorem{theorem}{Theorem}

\newtheorem{proposition}{Proposition}
\newtheorem{lemma}{Lemma}
\newtheorem{corollary}{Corollary}

\theoremstyle{remark}
\newtheorem{remark}{Remark}

\newcommand\pgfmathsinandcos[3]{%
  \pgfmathsetmacro#1{sin(#3)}%
  \pgfmathsetmacro#2{cos(#3)}%
}
\newcommand\LatitudePlane[3][current plane]{%
  \pgfmathsinandcos\sinEl\cosEl{#2} 
  \pgfmathsinandcos\sint\cost{#3} 
  \pgfmathsetmacro\yshift{\cosEl*\sint}
  \tikzset{#1/.style={cm={\cost,0,0,\cost*\sinEl,(0,\yshift)}}} %
}
\newcommand\NewLatitudePlane[4][current plane]{%
  \pgfmathsinandcos\sinEl\cosEl{#3} 
  \pgfmathsinandcos\sint\cost{#4} 
  \pgfmathsetmacro\yshift{#2*\cosEl*\sint}
  \tikzset{#1/.style={cm={\cost,0,0,\cost*\sinEl,(0,\yshift)}}} %
}
\newcommand\DrawLatitudeCircle[2][1]{
  \LatitudePlane{\angEl}{#2}
  \tikzset{current plane/.prefix style={scale=#1}}
  \pgfmathsetmacro\sinVis{sin(#2)/cos(#2)*sin(\angEl)/cos(\angEl)}
  \pgfmathsetmacro\angVis{asin(min(1,max(\sinVis,-1)))}
  \draw[current plane] (\angVis:1) arc (\angVis:-\angVis-180:1);
  \draw[current plane,dashed] (180-\angVis:1) arc (180-\angVis:\angVis:1);
}

\tikzset{%
  >=latex, 
  inner sep=0pt,%
  outer sep=2pt,%
  mark coordinate/.style={inner sep=0pt,outer sep=0pt,minimum size=3pt,
    fill=black,circle}%
}

\begin{document}

\title{The Fermi gerbe of Weyl semimetals}

\author[1]{Alan Carey}
\author[2]{Guo Chuan Thiang\thanks{thgchuan@gmail.com, ORCID 0000-0003-0268-0065}}
\affil[1]{Mathematical Sciences Institute, Australian National University, Canberra}
\affil[1]{School of  Mathematics and Applied Statistics, University of Wollongong, Wollongong, NSW}
\affil[2]{Beijing International Center for Mathematical Research, Peking University, Beijing, China}
\renewcommand\Authands{ and }

\date{\today}

\maketitle

\begin{abstract}
In the gap topology, the unbounded self-adjoint Fredholm operators on a Hilbert space have third homotopy group the integers. We realise the generator explicitly, using a family of Dirac operators on the half-line, which arises naturally in Weyl semimetals in solid-state physics. A ``Fermi gerbe'' geometrically encodes how discrete spectral data of the family interpolate between essential spectral gaps. Its non-vanishing Dixmier--Douady invariant protects the integrity of the interpolation, thereby providing topological protection of the Weyl semimetal's Fermi surface.

\end{abstract}

\section*{Introduction}
Bundle gerbes provide a useful geometric model for the third cohomology of a manifold, see \cite{Murray,MS,Hitchin} for an overview. They have been well-utilised in quantum field theory and string mathematics, see \cite{CMM} and \cite{CJM}, and have also begun to enter condensed matter physics in the setting of topological insulators \cite{Gawedzki, GT-gerbe}, as well as Floquet topological systems \cite{GomiTauber} and geometric phases \cite{Viennot}.

We will introduce the \emph{Fermi gerbe}, and apply it in an essential way in the hybrid context of \emph{Weyl semimetals} \cite{AMV, MT}, which are condensed matter realisations of the chiral anomaly and Weyl fermions. In 3D, the key experimental signature of a Weyl semimetal is found at the material boundary, where robust \emph{Fermi arcs} of boundary-localised ``edge states'' appear. Physically, the Fermi gerbe is a geometrical encoding of how conducting edge states interpolate across an insulating spectral gap of bulk states, and has broad applicability beyond Weyl semimetals. Its topological invariant --- the Dixmier--Douady (DD) class --- is precisely an obstruction to destroying the Fermi arc (or Fermi surface, in the higher-dimensional case) without closing the bulk insulating gap. As a minimal example, we prove that the 5D Weyl semimetal has a non-trivial Fermi gerbe.

A few technical points are in order before sketching how this works. In the operator norm topology, the space $\mathcal{F}^{\rm sa}_*$ of \emph{bounded} self-adjoint Fredholm operators which are neither essentially positive nor essentially negative, was shown by Atiyah--Singer \cite{AS} to be a classifying space for the odd $K$-theory functor. Furthermore, a non-contractible loop in $\pi_1(\mathcal{F}^{\rm sa}_*)=K^{-1}(S^1)\cong\ZZ$ is precisely one which exhibits spectral flow \cite{Phillips}. If one has a loop of \emph{unbounded} self-adjoint Fredholms, which is continuous in the so-called Riesz topology \cite{BLP}, the spectral flow is well-defined via the bounded transform. Likewise, the bundle gerbe construction for Riesz-continuous families of Dirac type operators can be found in \cite{CMFaddeev,CMUniversal}.

As it turns out, a more useful topology on the unbounded self-adjoint operators $\mathcal{C}^{\rm sa}$ is the weaker \emph{gap topology}, which is the topology of norm convergence of the resolvents $(A+i)^{-1}$ \cite{BLP, RS1}. Working with the gap topology allows us to circumvent restrictive requirements (e.g.\ compact resolvents and/or fixed domains) which are not satisfied in the physically relevant examples. In the gap topology, the unbounded self-adjoint Fredholm operators $\mathcal{CF}^{\rm sa}$, \emph{without restriction on the essential spectrum}, also provide a classifying space for $K^{-1}$ \cite{Joachim}, thus $\pi_n(\mathcal{CF}^{\rm sa})\cong\ZZ$ for all odd $n$. For instance, spectral flow of loops in $\mathcal{F}^{\rm sa}_*$ generalises in basically the same way to loops in $\mathcal{CF}^{\rm sa}$. We are specifically interested in $\pi_3(\mathcal{CF}^{\rm sa})\cong\ZZ$, and will prove that the natural 3-sphere of \emph{quaternionic} half-line Dirac operators generates $\pi_3(\mathcal{CF}^{\rm sa})$.

The above operator family is even naturally derived from the Weyl Hamiltonian modelling a Weyl semimetal. 
In 3D, the key observation of \cite{GCT} is that the Weyl Hamiltonian on a half-space is assembled from non-contractible loops (in $\mathcal{CF}^{\rm sa}$) of Dirac operators on the Euclidean half-line $\RR_+$. These loops have homotopy invariant spectral flows, resulting in the experimentally robust Fermi arcs. Similarly, the 5D Weyl Hamiltonian is assembled from 3-spheres of quaternionic half-line Dirac operators. The latter has Fermi gerbe with non-trivial DD-invariant, leading to ``higher-dimensional spectral flow'' prefiguring the topologically protected Fermi surface (see Fig.\ \ref{fig:Fermi.sphere}).

On the issue of experimental accessibility of the Fermi gerbe concept, let us mention that ``synthetic dimensions'' can be introduced in optical systems, to simulate Hamiltonians which formally describe systems in dimensions greater than 3 \cite{Ozawa}. In a follow-up work \cite{GT-gerbe}, it is also shown that a `Real' version of the Fermi gerbe encodes the Dirac cone edge spectrum of 2D and 3D topological insulators with time-reversal symmetry. Finally, we mention recent theoretical \cite{Palumbo} and experimental work \cite{Tan}, which study fictitious ``tensor monopoles'' in 4D with associated DD-invariant.

\medskip

{\bf Notation}: For a self-adjoint $A$, we write $\sigma_{\rm d}(A), \sigma_{\rm ess}(A)$ and $\sigma(A)$ for its discrete spectrum, essential spectrum, and (full) spectrum, respectively.

\section{Warm-up: Spectral flow of half-line Dirac operators and 3D Weyl semimetals}\label{sec:warm.up}
The Dirac operator on the Euclidean half-line $\RR_+$ (i.e.\ $z\geq 0$), with mass $\rho\geq 0$, is the operator
\begin{equation}
{D}(\rho)=\begin{pmatrix}-i\frac{d}{dz} & \rho \\ \rho & i\frac{d}{dz}\end{pmatrix},\label{eqn:Dirac.massive}
\end{equation}
which is formally self-adjoint on the smooth sections with compact support away from the boundary, $C_0^\infty(\RR_+; \CC\oplus\CC)\subset L^2(\RR_+; \CC\oplus\CC)$. For each $\omega\in{\rm U}(1)$, there is a self-adjoint extension ${D}(\rho;\omega)$, defined by the boundary condition 
\begin{equation*}
\psi(z=0)\propto\begin{pmatrix}1 \\ \omega\end{pmatrix}.
\end{equation*}
In \cite{GCT}, it was shown that for any $\rho>0$, the operator loop $\{{D}(\rho;\omega)\}_{\omega\in{\rm U}(1)}$ is gap-continuous in $\mathcal{CF}^{\rm sa}$, with common essential spectrum
\begin{equation*}
\sigma_{\rm ess}({D}(\rho;\omega))=(-\infty,-\rho]\cup[\rho,\infty),\qquad \forall\;\omega\in{\rm U}(1).
\end{equation*}
In the essential spectral gap $(-\rho,\rho)$, the discrete spectrum is
\begin{equation*}
\sigma_{\rm d}({D}(\rho;\omega))=\begin{cases}\rho\cdot{\rm Re}(\omega),\qquad\;\; {\rm Im}(\omega)<0,\\ \emptyset,\qquad\qquad\qquad{\rm Im}(\omega)\geq 0.\end{cases}
\end{equation*}

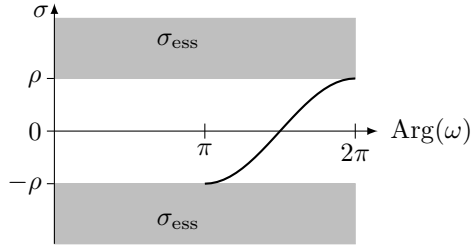
\begin{figure}[h]
\begin{center}
\begin{tikzpicture}
\draw[->] (-0.1,0) -- (4.3,0);
\node [right] at (4.4,0) {${\rm Arg}(\omega)$};
\node [left] at (0,1.6) {$\sigma$};
\node [left] at (-0.10,0.7) {$\rho$};
\node [left] at (-0.10,-0.7) {$-\rho$};
\node [left] at (-0.10,0) {$0$};
\node [below] at (2,-0.1) {$\pi$};
\node [below] at (4,-0.1) {$2\pi$};
\draw[-] (-0.1,0.7) -- (0.1,0.7);
\draw[-] (-0.1,-0.7) -- (0.1,-0.7);
\draw[-] (2,-0.1) -- (2,0.1);
\draw[-] (4,-0.1) -- (4,0.1);
\filldraw[lightgray] (0,0.7) rectangle (4,1.5);
\filldraw[lightgray] (0,-1.5) rectangle (4,-0.7);
\node [left] at (2,1.2) {$\sigma_{\rm ess}$};
\node [left] at (2,-1.2) {$\sigma_{\rm ess}$};
\draw[thick,domain=2:4] plot (\x,{0.7*cos(0.5*\x*pi r)});
\draw[->] (0,-1.5) -- (0,1.7);
\end{tikzpicture}
\end{center}
\caption{Spectrum of half-line Dirac operator $D(\rho,\omega)$ with mass $\rho>0$, as a function of the boundary condition $\omega\in{\rm U}(1)$. The essential spectrum is the grey shaded region, while the black curve joins up the eigenvalues of $D(\rho,\omega)$. Overall, the family $\{D(\rho,\omega)\}_{\omega\in{\rm U}(1)}$ has spectral flow 1.}\label{fig:basic.spectral.flow}
\end{figure}

\begin{proposition}[\cite{GCT}]
For any $\rho>0$, the loop $\{{D}(\rho;\omega)\}_{\omega\in{\rm U}(1)}$ of massive half-line Dirac operators represents a generator of $\pi_1(\mathcal{CF}^{\rm sa})\cong\ZZ$.
\end{proposition}
\begin{proof}
The isomorphism $\pi_1(\mathcal{CF}^{\rm sa})\cong\ZZ$ is given by the spectral flow (across 0) \cite{BLP}, and for $\{{D}(\rho;\omega)\}_{\omega\in{\rm U}(1)}$, the spectral flow is $+1$ (see Fig.\ \ref{fig:basic.spectral.flow}).
\end{proof}

{\bf Application to 3D Weyl semimetals.}
In the context of Weyl semimetals occupying the upper-half 3D Euclidean space subject to some boundary condition $\omega_0\in{\rm U}(1)$, the $(\rho,\omega)$ occur as polar coordinates for the 2D momentum space $\wh{\RR}^2$ parallel to the boundary surface $z=0$. Fourier transforming the half-space Weyl Hamiltonian along the boundary surface directions converts it into $\int^\oplus_{(\rho,\omega)\in\wh{\RR}^2}{D}(\rho;\overline{\omega}\omega_0)$, and ${D}(\rho;\overline{\omega}\omega_0)$ has a zero eigenvalue exactly when $(\rho,\omega)$ lies on the ray ${\rm Arg}(\omega\overline{\omega_0})=\frac{\pi}{2}$. This ray connecting the origin of $\wh{\RR}^2$ to infinity is the \emph{Fermi arc}, and is ``topologically protected'' due to homotopy invariance of the spectral flows of $\{{D}(\rho;\overline{\omega})\}_{(\rho,\omega)\in\wh{\RR}^2{\setminus}\{0\}}$. Further details can be found in \cite{GCT}.

\section{Spectrum of quaternionic Dirac operators on the half-line}\label{sec:quaternionic.Dirac}
The \emph{quaternionic} Dirac operator on the half-line, with mass $\rho=1$, is defined to be the operator
\begin{equation}
\slashed{D}=\underbrace{\begin{pmatrix}-i\frac{d}{dz} & 0 & 1 & 0 \\ 0 & -i\frac{d}{dz} & 0 & 1 \\ 1 & 0 & i\frac{d}{dz} & 0 \\ 0 & 1 & 0 & i\frac{d}{dz}\end{pmatrix}}_{\rm complex\;notation}\equiv\underbrace{\begin{pmatrix}-i\frac{d}{dz}  & 1 \\ 1 & i\frac{d}{dz}\end{pmatrix}}_{\parbox{5em}{quaternionic\\ notation}},\label{eqn:Dirac.quaternion}
\end{equation}
formally self-adjoint on $C_0^\infty(\RR_+;(\CC^2\oplus\CC^2))\cong C_0^\infty(\RR_+;(\HH\oplus\HH))$. When switching to quaternionic notation, we identify each $\CC^2$ as a quaternionic vector space $\HH$, using the quaternionic structure $\Theta=\begin{pmatrix} 0 & -1 \\ 1 & 0\end{pmatrix}\circ\kappa$, where $\kappa$ denotes complex conjugation. Operators on $\CC^4$ commuting with $\Theta$ are then $2\times 2$ matrices with quaternion entries, and each quaternion entry is itself represented as a $2\times 2$ complex matrix. More details of conventions for quaternions can be found in Appendix \ref{sec:appendix.quaternion}.

The unbounded part ${\rm diag}(-i\frac{d}{dz}, -i\frac{d}{dz}, i\frac{d}{dz}, i\frac{d}{dz})$ of $\slashed{D}$ has two-dimensional $\pm i$ eigenspaces (the \emph{deficiency subspaces}), spanned over $\CC$ by $\mathsf{e_1}e^{-z}, \mathsf{e_2} e^{-z}$ and $\mathsf{e_3}e^{-z}, \mathsf{e_4} e^{-z}$, respectively, where $\mathsf{e_j}$ denote the standard basis vectors of $\CC^4$. Self-adjoint extensions of $\slashed{D}$ are fully parametrised by a unitary mapping between the deficiency subspaces \cite{RS2}. With respect to the above choice of basis, we may label such a map by elements of ${\rm U}(2)$.

Let us restrict to ${\rm Sp}(1)\cong{\rm SU}(2)\subset{\rm U}(2)$, in order to be compatible with $\CC^4\cong\HH^2$ as a quaternionic vector space. In terms of boundary conditions, the self-adjoint extension $\slashed{D}(q)$ of $\slashed{D}$ specified by $q\in{\rm Sp}(1)$ has domain
\begin{equation}
{\rm Dom}(\slashed{D}(q))=\left\{\psi\in H^1(\RR_+;\CC^4)\;:\; \psi(0)=\begin{pmatrix}u \\ qu\end{pmatrix},\;\;u\in\CC^2\right\},\label{eqn:Dirac.quaternion.bc}
\end{equation}
where $H^1\subset L^2$ denotes the Sobolev space with $L^2$ weak first derivatives. 

\begin{lemma}\label{lem:continuity.of.family}
For $q\in{\rm Sp}(1)\cong S^3$, let $\slashed{D}(q)$ be the self-adjoint extension of Eq.\ \eqref{eqn:Dirac.quaternion} determined by the boundary condition Eq.\ \eqref{eqn:Dirac.quaternion.bc}. Then the operator family $\{\slashed{D}(q)\}_{q\in{\rm Sp}(1)}$ is gap-continuous. 
\end{lemma}
\begin{proof}
Conjugation by the norm-continuous unitary family, $q\mapsto \begin{pmatrix} 1 & 0 \\ 0 & \overline{q}\end{pmatrix}$
converts $\{\slashed{D}(q)\}_{q\in{\rm Sp}(1)}$ into the family 
\begin{equation}
{\rm Sp}(1)\ni q\mapsto \begin{pmatrix}-i\frac{d}{dz} & q \\ \overline{q} & i\frac{d}{dz}\end{pmatrix},\qquad \psi(0)=\begin{pmatrix}u \\ u\end{pmatrix},\;\;u\in\CC^2.\label{eqn:conjugated.family}
\end{equation}
All operators in the transformed family have the same domain and depend on $q$ only through the bounded off-diagonal mass term. Thus, the transformed family is gap-continuous, as is the original family in question.
\end{proof}
\begin{lemma}\label{lem:essential.spec}
For each $q\in{\rm Sp}(1)$, the self-adjoint operator $\slashed{D}(q)$ is Fredholm with essential spectrum
\begin{equation*}
\sigma_{\rm ess}(\slashed{D}(q))=(-\infty,-1]\cup[1,\infty),\qquad \forall q\in{\rm Sp}(1).
\end{equation*}
\end{lemma}
\begin{proof}
Due to finite deficiency indices, the essential spectrum does not depend on the choice of self-adjoint extension \cite{RS2}, so it suffices to take $q=1$. In that case, we can unitarily conjugate $\slashed{D}(1)$ into
\begin{equation*}
\frac{1}{\sqrt{2}}\begin{pmatrix}1 & 1 \\ i & -i \end{pmatrix}\begin{pmatrix}-i\frac{d}{dz} & 1 \\ 1 & i\frac{d}{dz}\end{pmatrix}\frac{1}{\sqrt{2}}\begin{pmatrix}1 & -i \\ 1 & i\end{pmatrix}=\begin{pmatrix}1 & -\frac{d}{dz} \\ \frac{d}{dz} & -1\end{pmatrix},
\end{equation*}
with transformed boundary condition $\psi(0)=\frac{1}{\sqrt{2}}\begin{pmatrix}1 & 1 \\ i & -i\end{pmatrix}\begin{pmatrix}u \\ u\end{pmatrix}=\begin{pmatrix}\sqrt{2} u \\ 0\end{pmatrix}$. The latter is a Dirichlet condition in (only) the second component, and we may apply the odd Fourier transform there. The symbol of $\begin{pmatrix}1 & -\frac{d}{dz} \\ \frac{d}{dz} & -1\end{pmatrix}$ is
$\begin{pmatrix}1 & -ip \\ ip & -1\end{pmatrix},\;p\in\wh{\RR}$, with eigenvalues $\pm\sqrt{p^2+1}$. Then
\begin{align*}
\sigma(\slashed{D}(1))=\cup_{p\in\wh{\RR}} \{-\sqrt{p^2+1},\sqrt{p^2+1}\}&=(-\infty,-1]\cup[1,\infty)\\
&=\sigma_{\rm ess}(\slashed{D}(1))
=\sigma_{\rm ess}(\slashed{D}(q)).
\end{align*}
\end{proof}

It follows from Lemmas \ref{lem:continuity.of.family} and \ref{lem:essential.spec} that $\{\slashed{D}(q)\}_{q\in{\rm Sp}(1)}$ defines a class in $\pi_3(\mathcal{CF}^{\rm sa})$. An indication that this homotopy class might be non-trivial, is given by a computation of the discrete spectrum. For this, we recall (Appendix \ref{sec:appendix.quaternion}) that $q\in{\rm Sp}(1)$ can be written as a sum of real and imaginary parts, $q=q_r+i\vect{q}\cdot\vect{\sigma}\equiv q_r+\sum_{j=1}^3 q_j\sigma_j$, where $q_r^2+|\vect{q}|^2=1$ and $\sigma_j$ are the Pauli matrices,
\begin{equation}
\sigma_1=\begin{pmatrix} 0 & 1 \\ 1 & 0 \end{pmatrix},\quad \sigma_2=\begin{pmatrix} 0 & -i \\ i & 0 \end{pmatrix},\quad\sigma_3=\begin{pmatrix} 1 & 0 \\ 0 & -1 \end{pmatrix}.\label{eqn:Pauli.matrices}
\end{equation}
\begin{proposition}\label{prop:discrete.spectrum}
For $q=q_r+i\vect{q}\cdot\vect{\sigma}\in{\rm Sp}(1)$, the discrete spectrum of $\slashed{D}(q)$ is
\begin{equation*}
\sigma_{\rm d}(\slashed{D}(q))=\begin{cases}\{q_r\},\qquad |q_r|<1,\\ \emptyset,\qquad\quad\;\, |q_r|=1.\end{cases}
\end{equation*}
In the former case, the eigenspace of $\slashed{D}(q)$ is spanned by the eigenfunction $\begin{pmatrix} u^-\\ q u^-\end{pmatrix} e^{-|\vect{q}|z}$,
where $u^-\in\CC^2$ solves $(\vect{q}\cdot\vect{\sigma})u^-=-|\vect{q}|u^-$.
\end{proposition}
\begin{proof}
By elliptic regularity, we are looking for smooth (strong) solutions to the eigenvalue problem,
\begin{equation}
\slashed{D}(q)\cdot\psi\equiv \begin{pmatrix}-i\frac{d}{dz} & 1 \\ 1 & i\frac{d}{dz}\end{pmatrix}\psi=\lambda\psi,\label{eqn:eigenvalue.equation0}
\end{equation}
for eigenvalues $\lambda\in(-1,1)$. So the eigenfunctions $\psi$ should be of the form $\psi=\psi_a : z\mapsto\begin{pmatrix}u \\ qu\end{pmatrix}e^{-az}, u\in\CC^2$, for some $a\in\CC$, whence Eq.\ \eqref{eqn:eigenvalue.equation0} reduces to the algebraic equation
\begin{equation*}
\begin{pmatrix}ia & 1 \\ 1  & -ia\end{pmatrix}\begin{pmatrix}u \\ qu\end{pmatrix}=\begin{pmatrix}(ia+q) u \\ q(\overline{q}-ia) u\end{pmatrix}=\lambda\begin{pmatrix}u \\ qu\end{pmatrix}.
\end{equation*}
Equivalently, $u\in\CC^2$ solves the coupled pair of eigenvalue equations
\begin{equation*}
(ia+q)u\equiv(q_r+i(a+\vect{q}\cdot\vect{\sigma}))=\lambda u,\qquad(\overline{q}-ia)u\equiv(q_r-i(a+\vect{q}\cdot\vect{\sigma}))u=\lambda u, 
\end{equation*}
whose solution requires $\lambda=q_r$, together with $(a+\vect{q}\cdot\vect{\sigma})u=0$.
The last condition is solved by taking $a=\mp |\vect{q}|$ and $u=u^\pm$ a positive/negative spinor of $\vect{q}\cdot\vect{\sigma}$, i.e.\ $(\vect{q}\cdot\vect{\sigma})u^\pm=\pm|\vect{q}|u^\pm$ (note that $ |q_r|=|\lambda|<1\Leftrightarrow\vect{q}\neq 0$). We require $a>0$ for normalisability of $\psi_a$, so we must take the lower signs throughout.
\end{proof}

\begin{remark}
Observe that every $\lambda$ in the essential spectral gap $(-1,1)$ is attained as an eigenvalue of $\slashed{D}(q)$ for some $q$. Thus the family $\{\slashed{D}(q)\}_{q\in{\rm Sp}(1)}$ has the \emph{gap-filling property}: the discrete spectra ``flows across the essential spectral gap'', see Fig.\ \ref{fig:Fermi.sphere}. However, this cannot be attributed to the usual spectral flow along any closed path (compare Fig.\ \ref{fig:basic.spectral.flow}), since the parameter space ${\rm Sp}(1)$ is simply connected. As such, we need to find some other homotopy invariant of the family $\{\slashed{D}(q)\}_{q\in{\rm Sp}(1)}$ to demonstrate that its gap-filling property is  ``topologically protected'', rather than spurious.
\end{remark}
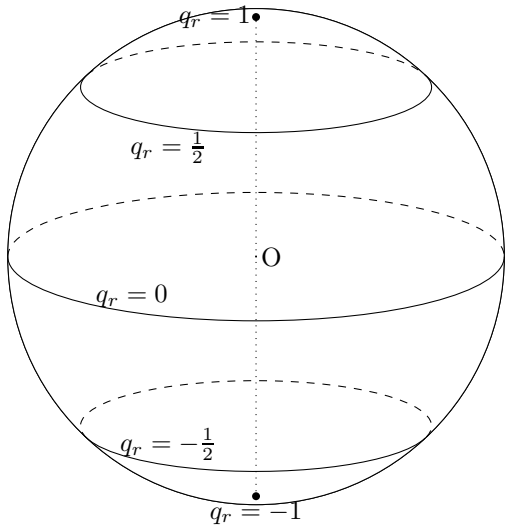
\begin{figure}[h]
\begin{center}

\begin{tikzpicture} 
\def\R{3.3 } 
\def\angEl{15} 
\def\angAz{-20} 
\filldraw[ball color=white] (0,0) circle (\R);
\filldraw[fill=white] (0,0) circle (\R);

\foreach \t in {-45,0,45} { \DrawLatitudeCircle[\R]{\t} }

\pgfmathsetmacro\H{\R*cos(\angEl)} 
\coordinate (0) at (0,0);
\node[circle,draw,black,scale=0.3] at (0,0) {};
\draw[right] node at (0,0){O};
\coordinate[mark coordinate] (N) at (0,\H);
\draw[left] node at (0,\H){$q_r=1$};
\coordinate[mark coordinate] (S) at (0,-\H);
\draw[right,below] node at (0,-\H){$q_r=-1$};
\draw[dotted, black](N)--(S);

\NewLatitudePlane[planeP]{\R}{\angEl}{45};
\path[planeP] (-120:\R) coordinate (P);
\draw[below] node at (P){$q_r=\frac{1}{2}$};

\NewLatitudePlane[equator]{\R}{\angEl}{00};
\path[equator] (-120:\R) coordinate (Pprime);
\draw[above] node at (Pprime){$q_r=0$};

\NewLatitudePlane[planePprimeprime]{\R}{\angEl}{-45};
\path[equator] (-120:\R) coordinate (Pprimeprime);
\draw[above] node at (Pprimeprime){$q_r=-\frac{1}{2}$};



\end{tikzpicture}

\end{center}
\caption{The sphere (one dimension suppressed) represents the unit quaternions ${\rm Sp}(1)$, with vertical axis labelling the real part $q_r$ of a quaternion $q$. For boundary condition $q\in {\rm Sp}(1)$ , the quaternionic half-line Dirac operator $\slashed{D}(q)$ has eigenvalue $\lambda$ inside the essential spectral gap $(-1,1)$, precisely when $q_r=\lambda$ (the 2-sphere at latitude $\lambda$). For the Weyl Hamiltonian on the 5D half-space, view ${\rm Sp}(1)\subset\wh{\RR}^4$ as the unit 3-sphere of boundary momenta $\vect{p}_\parallel=(p_1,p_2,p_3,p_4)$, with $p_1$ aligned vertically. For boundary condition $\varGamma=1$ and Fermi level $\mu$, the half-space Weyl Hamiltonian has Fermi surface the hyperplane $q_r\equiv p_1=\mu$.}\label{fig:Fermi.sphere}
\end{figure}

\begin{remark}
A similar eigenvalue computation to that of Prop.\ \ref{prop:discrete.spectrum} was carried out in \S III of \cite{HWK}, but our use of quaternions to simplify the analysis and obtain the eigenfunctions explicitly in terms of $q$, is new and especially useful for what follows.
\end{remark}

\section{Fermi gerbe construction}
When describing electronic properties of materials, one often encounters a family $F$ of self-adjoint Hamiltonians, parametrised by (some subset of) momentum space. When the family shares some gap (around a given real number called the \emph{Fermi  level}, which we set to 0) in the essential spectrum, we have an insulator ``in the bulk''. The example $\{\slashed{D}(q)\}_{q\in{\rm Sp}(1)}$ in \S\ref{sec:quaternionic.Dirac} comes from a model of Weyl semimetals, see Remark \ref{rem:Weyl.semimetal}.

We will construct a gerbe $\mathcal{G}_F$ using the \emph{discrete} spectral data of a family $F$ that has a common essential spectral gap. Physically, the discrete spectrum comprises boundary-localised ``edge states''. There has been tremendous interest in ``topological'' Hamiltonian families whose edge states interpolate between the bulk states in a ``topologically protected'' manner. 
Intuitively, $\mathcal{G}_F$ encodes this spectral interpolation, and it is trivialisable whenever the interpolation is ``breakable''. Thus, a family with non-trivialisable $\mathcal{G}_F$ will have ``topologically protected edge states''. 

The construction of $\mathcal{G}_F$ is an adaptation of that in \cite{CMFaddeev,CMM,PS}, which was motivated by anomalies in the quantisation of chiral fermions coupled to vector potentials. There, the parameter space of connections is contractible, so that an anomaly (non-trivial gerbe) arises only after passing to the moduli space (quotient by gauge transformations). We work in the more general setting of gap-continuous unbounded self-adjoint Fredholm families, which admits our minimal example of an anomalous Dirac operator family $\slashed{D}$, whose gerbe $\mathcal{G}_\slashed{D}$ is constructed directly over a 3-sphere. Here, $\slashed{D}$ is short for the family $\{\slashed{D}(q)\}_{q\in{\rm Sp}(1)}$ analysed in \S\ref{sec:quaternionic.Dirac}.

\vspace{1em}
 
  Let $X$ be a paracompact topological space and $F:X\rightarrow \mathcal{CF}^{\rm sa}$ be a gap-continuous family.  In addition we assume that there is an interval in the real numbers containing zero such that there is no essential spectrum in this interval for all operators in the family. Following \cite{AS}, we may homotope to a situation where the common spectral gap for the family is the interval $(-1,1)$.
 
 For completeness, we sketch the argument.  We let $p$ be the map from the bounded operators on a separable Hilbert space $\mathcal H$ to the Calkin algebra (the quotient of the bounded operators by the compact operators).  Recalling that the gap topology is the topology of norm resolvent convergence, we see that the map
 $x\mapsto || p((i+F(x))^{-1})||=: \rho(x)$ is continuous.  Then, $\inf|\sigma_{\rm ess}(F(x))|=\frac{\sqrt{1-\rho(x)^2}}{\rho(x)}$ is smaller than $1$ when $\rho(x)>\frac{1}{\sqrt{2}}$, in which case we replace $F(x)$ by \mbox{$F(x) \rho(x)(1-\rho(x)^2)^{-1/2}$}. This does not change any of the topological considerations that follow.
 Thus, assume henceforth that $(-1,1)$ is a common essential spectral gap and so each $F(x)$ may have isolated, finite-multiplicity eigenvalues (i.e.\ discrete spectrum) inside $(-1,1)$, possibly accumulating only at $\pm 1$.

For a gap-continuous self-adjoint family, membership in the spectrum is a closed condition (Theorem VIII.23, \cite{RS1}). So with $F:X\rightarrow \mathcal{CF}^{\rm sa}$ as above, we may define for each $\lambda\in (-1,1)$, the open set \begin{equation*}
U_\lambda=\{x\in X\,:\,\lambda\not\in\sigma(F(x))\}=\{x\in X\,:\,\lambda\not\in\sigma_{\rm d}(F(x))\}.
\end{equation*}
Then, $\mathcal{U}=\{U_\lambda\}_{-1<\lambda<1}$ is the \emph{standard open cover} for $X$ determined by $F$. As recalled in Appendix \ref{appendix:bundle.gerbes}, a gerbe may be specified locally by assigning Hermitian ``transition line bundles'' $\mathcal{L}_{U_\lambda, U_\mu}\rightarrow U_\lambda\cap U_\mu$ to each ordered pair in the open cover. If $\lambda<\mu$, take $\mathcal{L}_{U_\lambda, U_\mu}\rightarrow U_\lambda\cap U_\mu$ to be the determinant line bundle (i.e.\ top exterior power) of the eigenspaces for eigenvalues lying within $(\lambda,\mu)$. Note that only discrete spectra away from accumulation points are involved, so there are only finitely many eigenvalues (counted with multiplicity) in the interval $(\lambda,\mu)$. If $\lambda>\mu$, we assign the dual of the above, i.e.\ $\mathcal{L}^*_{U_\lambda,U_\mu}$, to $U_\lambda\cap U_\mu$. For $\lambda=\mu$, we assign the trivial line bundle $\underline{\CC}$. Note that if $x\in U_\lambda\cap U_\mu$ is such that $F(x)$ has no eigenvalues in $(\lambda,\mu)$, then the complex line assigned to $x$ is the canonical copy of $\CC$. 

We sketch why the bundle $\mathcal{L}_{U_\lambda, U_\mu}$ is locally trivial. Fix $x_0$ in $X$ and let $\chi_{\lambda\mu}$ be the characteristic function of the interval $(\lambda,\mu)$.
Proposition 2.10(b) of \cite{BLP} establishes the existence of a neighbourhood of $x_0$ on which $\chi_{\lambda\mu}(F(x))$ is a continuous family of finite rank projections.
It follows that there is a neighbourhood of $x_0$ on which the projections from the range of  $\chi_{\lambda\mu}(F(x))$ to the range of  $\chi_{\lambda\mu}(F(x_0))$ form a continuous family of linear isomorphisms. These isomorphisms induce isomorphisms of determinant lines 
which in turn may be used to construct a local trivialisation.

On triple intersections $U_\lambda\cap U_\mu\cap U_\nu$, one verifies that the tensored line bundles $\mathcal{L}_{\lambda,\mu}\otimes\mathcal{L}_{\mu,\nu}\otimes\mathcal{L}_{\nu,\lambda}$ are canonically trivial. Equivalently, $\mathcal{L}_{\lambda,\mu}\otimes\mathcal{L}_{\mu,\nu}$ and $\mathcal{L}_{\lambda,\nu}$ are canonically isomorphic. For instance, suppose $\lambda<\mu<\nu$. Then, when restricted to those $x\in U_\lambda\cap U_\mu\cap U_\nu$ for which $F(x)$ has eigenvalues in $(\lambda,\mu)$, the three line bundle factors are $\mathcal{L}_{\lambda,\mu}$, $\underline{\CC}$, and $\mathcal{L}^*_{\lambda,\mu}$, and they tensor to $\underline{\CC}$; a similarly trivial tensor product occurs for the other spectral intervals $(-1,\lambda), (\mu,\nu)$ and $(\nu,1)$.
The above assignment of line bundles thereby defines a gerbe $\mathcal{G}_F$ associated with the self-adjoint Fredholm family $F$, which we call the \emph{Fermi gerbe}.

\vspace{1em}
Quite generally, any gerbe $\mathcal{G}$ over $X$ has a \emph{Dixmier--Douady} (DD) invariant in $H^3(X,\ZZ)$ \cite{Br,Murray}. In short, the line bundle assignment to double overlaps of a given open cover, gives a \v{C}ech 2-cocycle with coefficients in the sheaf of continuous ${\rm U}(1)$-valued functions on $X$. Its cohomology class in $H^2(X,\underline{{\rm U}(1)})\cong H^3(X,\ZZ)$ is the direct limit over refinements of open covers of $X$.

\vspace{1em}
The DD-invariant of $\mathcal{G}_F$ has physical significance as an obstruction to maintaining a common spectral gap throughout the family $F$:
\begin{lemma}\label{lem:gap.vanishing}
Let $F:X\rightarrow \mathcal{CF}^{\rm sa}$ be a gap-continuous family with a common essential spectral gap $(-1,1)$. Suppose there exists some $\lambda_0\in(-1,1)$ such that $\lambda_0\not\in\sigma(F(x)), \forall\, x\in X$. Then, the DD-invariant of the Fermi gerbe $\mathcal{G}_F$ vanishes.
\end{lemma}
\begin{proof}
For such a family, the open set $U_{\lambda_0}$ already covers all of $X$. The cover comprising just $U_{\lambda_0}$ refines the standard one, and the \v{C}ech 2-cocycle is manifestly trivial on this cover. Passing to the direct limit over refinements, the DD-invariant likewise trivialises.
\end{proof}

For a general family $F:X\rightarrow \mathcal{CF}^{\rm sa}$, it is quite a challenge to compute the DD-invariant of its Fermi gerbe $\mathcal{G}_F$. Nevertheless, we can do so for $\mathcal{G}_\slashed{D}$.

\begin{theorem}\label{thm:DD.family}
For the quaternionic half-line Dirac operator family $\{\slashed{D}(q)\}_{q\in{\rm Sp}(1)}$, the associated Fermi gerbe $\mathcal{G}_{\slashed{D}}$ has Dixmier--Douady invariant a generator of $H^3({\rm Sp}(1),\ZZ)\cong \ZZ$.
\end{theorem}
\begin{proof}
For any given $\lambda\in (-1,1)$, Prop.\ \ref{prop:discrete.spectrum} shows that $\slashed{D}(q)$ has eigenvalue $\lambda$ precisely when $q$ belongs to the 2-sphere 
\begin{equation*}
S^2_\lambda:=\{q=q_r+i\vect{q}\cdot\vect{\sigma}\in{\rm Sp}(1)\,:\,q_r=\lambda\},
\end{equation*}
which we may parametrise by $\vect{q}$.
Furthermore, the bundle of eigenvectors over $S^2_\lambda$ may be identified with the negative eigenbundle of $\{\vect{q}\cdot\vect{\sigma}\}_{\vect{q}\in S^2_\lambda}$, i.e.\ a Hopf line bundle (this is the Bloch sphere identification $S^2\cong \CC\PP^1$). It is convenient to pass to the open cover of ${\rm Sp}(1)$, comprising
\begin{align*}
V&=\{q\in {\rm Sp}(1)\,:\, \sigma_{\rm d}(\slashed{D}(q)))\subset (-\epsilon,1)\},\\
W&=\{q\in {\rm Sp}(1)\,:\, \sigma_{\rm d}(\slashed{D}(q))\subset (-1,\epsilon)\},
\end{align*}
where $\epsilon>0$ is some small number. Note that $V\cup W={\rm Sp}(1)$, with $V\subset U_{-\epsilon}$ and $W\subset U_{\epsilon}$, so $\{V,W\}$ is a refinement of the standard cover $\mathcal{U}$. The Fermi gerbe is trivialised over $V$, and over $W$, but globally, it is clutched over the equatorial band $V\cap W=\cup_{\lambda\in (-\epsilon,\epsilon)}S^2_{\lambda}\subset U_{-\epsilon}\cap U_{\epsilon}$, which is a slightly thickened version of $S^2_0$. By definition, the Fermi gerbe $\mathcal{G}_\slashed{D}$ assigns the Hopf line bundle to $V\cap W$, and the latter has Chern class generating $H^2(S^2,\ZZ)\cong\ZZ$. This clutching construction shows that the DD-invariant of $\mathcal{G}_\slashed{D}$ generates $H^3({\rm Sp}(1),\ZZ)\cong\ZZ$.

\end{proof}

\begin{corollary}\label{cor:homotopy.generator}
The quaternionic half-line Dirac operator family $\{\slashed{D}(q)\}_{q\in{\rm Sp}(1)}$ represents a generator of $\pi_3(\mathcal{CF}^{\rm sa})\cong\ZZ$.
\end{corollary}
\begin{proof}
In $\pi_3(\mathcal{CF}^{\rm sa})$, the composition can be taken to be via direct sum of operators (and invoking Kuiper's theorem). Then, the operation of taking the DD-invariant of the Fermi gerbe of a family $F:S^3\rightarrow\mathcal{F}^{\rm sa}$ is a surjective homomorphism $\pi_3(\mathcal{CF}^{\rm sa})\rightarrow H^3(S^3,\ZZ)$ by Thm.\ \ref{thm:DD.family}, thus an isomorphism.
\end{proof}

\section{5D Weyl Hamiltonian and topological Fermi surface}\label{sec:physics}
In this section, we explain how the anomalous family $\{\slashed{D}(q)\}_{q\in{\rm Sp}(1)}$ arises from the 5D Weyl Hamiltonian (which is a Euclidean space Dirac operator).

Recall the spinor representation of ${\rm Spin}(5)\cong{\rm Sp}(2)$ on $\CC^4\cong\HH^2$. By definition \cite{LM}, an orthonormal basis $e_1,\ldots, e_5$ of Euclidean $\RR^5$ is ``quantised'' into generators of the Clifford algebra $Cl_{0,5}$, satisfying the relation $e_ie_j+e_je_i=2\delta_{ij}$. In the spinor representation, these quantised $e_i$ are represented as Hermitian $4\times 4$ Dirac matrices $\gamma_i$, with the chirality element $\gamma_1\gamma_2\gamma_3\gamma_4\gamma_5$ acting as a scalar $-1$. Concretely, in terms of Pauli matrices, Eq.\ \eqref{eqn:Pauli.matrices}, we could choose\footnote{With minor relabelling, our convention here matches Eq.\ 5.1 of \cite{ASSS}, and Eq.\ 4.2 of \cite{MT}.}
\begin{equation}
\gamma_1=\sigma_1\otimes 1_2,\quad \gamma_2=\sigma_2\otimes\sigma_3,\quad\gamma_3=-\sigma_2\otimes\sigma_2,\quad \gamma_4=\sigma_2\otimes \sigma_1,\quad \gamma_5=\sigma_3\otimes 1_2.\label{eqn:Dirac.matrices}
\end{equation}
That the spinor representation is quaternionic can be explicitly seen by choosing $\Theta=(1_2\otimes -i\sigma_2)\circ\kappa$ with $\kappa$ denoting complex conjugation, and checking that $\Theta$ commutes with each $\gamma_i$ and thus the entire Clifford algebra representation.

The 5D \emph{Weyl Hamiltonian} is $H^{\rm W}=-i\nabla\cdot\vect{\gamma}$. More precisely, $H^{\rm W}$ is self-adjoint on the Sobolev space $H^1(\RR^5)\otimes\CC^4$, and Fourier transforming to momentum space $\wh{\RR}^5$ gives
\begin{equation*}
H^{\rm W}=-i\nabla\cdot\vect{\gamma}\cong\int^\oplus_{\vect{p}\in\wh{\RR}^5} \vect{p}\cdot\vect{\gamma}\equiv\int^\oplus_{\vect{p}\in\wh{\RR}^5} \sum_{j=1}^5 p_j\gamma_j=:\int^\oplus_{\vect{p}\in\wh{\RR}^5} Q(\vect{p}),
\end{equation*}
Here, each $Q(\vect{p})=\vect{p}\cdot\vect{\gamma}$ is quaternionic linear on $\CC^4\cong\HH^2$, so we may think of it as an element of $M_2(\HH)$. Operators of the form $Q(\vect{p})$ were called \emph{quadrupole Hamiltonians} in \cite{ASSS}.

\begin{remark}
The quaternionic structure $\Theta$ acts fiberwise over momentum space $\wh{\RR}^5$, but when Fourier transformed to position space $\RR^5$, it is actually the composition of a time-reversal operator $\mathsf{T}=(1_2\otimes -i\sigma_2)\circ\kappa$ with spatial inversion $\mathsf{P}:x\mapsto -x\in\RR^5$. That is, $\Theta$ is really a $\mathsf{PT}$ symmetry for $H^{\rm W}=-i\nabla\cdot\vect{\gamma}$. 
\end{remark}

Write $z$ for the fifth coordinate of $\RR^5$. We wish to define a half-space ($z\geq 0$) version of $H^{\rm W}$. As the momentum $p_5$ along $z$ will no longer be conserved, it is convenient to think of $\wh{\RR}^5=\wh{\RR}^4\oplus\wh{\RR}$ as $\HH\oplus\wh{\RR}$, with the boundary-parallel momentum $\vect{p}_\parallel=(p_1,p_2,p_3,p_4)\in\wh{\RR}^4$ repackaged as a single quaternion\footnote{For the real/imaginary decomposition of the quaternion $q=q_r+i\vect{q}\cdot\vect{\sigma}$ (see Appendix \ref{sec:appendix.quaternion}) of Eq.\ \eqref{eqn:quaternion.package}, we have $q_r=p_1$, and $\vect{q}=(p_4,-p_3,p_2)$.}, 
\begin{equation}
\wh{\RR}^4\ni \vect{p}_\parallel\longleftrightarrow q=\begin{pmatrix}p_1+ip_2 & -p_3+ip_4 \\ p_3+ip_4 & p_1-ip_2\end{pmatrix} \in \HH.\label{eqn:quaternion.package}
\end{equation}
Then the map $\vect{p}\mapsto Q(\vect{p})$ is concisely written as
\begin{equation*}
(q, p_5)\equiv(\vect{p}_\parallel,p_5)=\vect{p}\mapsto Q(\vect{p})=\vect{p}\cdot\vect{\gamma}=\begin{pmatrix} p_5 & \overline{q} \\ q & -p_5\end{pmatrix}\in M_{2}(\HH).
\end{equation*}
Undoing the Fourier transformation in the last variable, $z\rightarrow p_5$, we have
\begin{equation*}
H^{\rm W}=\int^\oplus_{\vect{p}_\parallel\in\wh{\RR}^4}\begin{pmatrix}-i\frac{d}{dz} & \overline{q} \\ q & i\frac{d}{dz}\end{pmatrix}=:\int^\oplus_{q\in\HH} h^{\rm W}(q),
\end{equation*}
where recognise each partially Fourier transformed $h^{\rm W}(q)$ as a quaternionic-linear Dirac operator on the line $\RR$ with ``mass term'' $q$ (cf.\ Eq.\ \eqref{eqn:Dirac.massive}); each $h^{\rm W}(q)$ is self-adjoint on $H^1(\RR; \HH^2)$.

We can now construct half-space versions of $H^{\rm W}$, by restricting each of the formal 1D operators $h^{\rm W}(q)$ to $C_0^\infty(\RR_+; \HH^2)$, and picking a boundary condition $\varGamma\in{\rm Sp}(1)$ to make them self-adjoint operators $\wt{h}^{\rm W}(q;\varGamma)$. Then, the self-adjoint half-space Weyl Hamiltonian $\wt{H}^{\rm W}$ with boundary condition $\varGamma\in{\rm Sp}(1)$ is decomposed as
\begin{align*}
\wt{H}^{\rm W}(\varGamma)&=\int^\oplus_{q\in\HH} \begin{pmatrix}-i\frac{d}{dz} & \overline{q} \\ q & i\frac{d}{dz}\end{pmatrix}=:\int^\oplus_{q\in\HH}\wt{h}^{\rm W}(q;\varGamma),\\
{\rm Dom}(\wt{h}^{\rm W}(q;\varGamma))&=\left\{\psi\in H^1(\RR_+;\HH^2)\;:\; \psi(0)=\begin{pmatrix}u \\ \varGamma u\end{pmatrix},\;u\in\HH\cong\CC^2\right\}.
\end{align*}
For $q$ in the unit 3-sphere ${\rm Sp}(1)$ of momentum space $\wh{\RR}^4$, we have in fact already encountered $\wt{h}^{\rm W}(q;\varGamma)$ in a different guise --- it is just a unitarily conjugated version of $\slashed{D}(q)$ (see Eq.\ \eqref{eqn:conjugated.family}). Thus, we have unitarily equivalent families
\begin{equation}
\{\wt{h}^{\rm W}(q;\varGamma)\}_{q\in{\rm Sp}(1)}\longleftrightarrow \{\slashed{D}(\overline{q}\varGamma)\}_{q\in{\rm Sp}(1)}\label{eqn:transformed.family}
\end{equation}
related through conjugation by $\left\{\begin{pmatrix} 1 & 0 \\ 0 & \overline{q} \end{pmatrix}\right\}_{q\in{\rm Sp}(1)}$. Since $\{\slashed{D}(\overline{q}\varGamma)\}_{q\in{\rm Sp}(1)}$ is just a reparametrisation of $\{\slashed{D}(q)\}_{q\in{\rm Sp}(1)}$ by a degree $-1$ homeomorphism of ${\rm Sp}(1)$, it follows immediately that:

\begin{corollary}\label{cor:Weyl.DD.invariant}[to Thm.\ \ref{thm:DD.family}]
For the half-space Weyl Hamiltonian $\wt{H}^{\rm W}(\varGamma)$ with boundary condition $\varGamma\in{\rm Sp}(1)$, let $\{\wt{h}^{\rm W}(q;\varGamma)\}_{q\in{\rm Sp}(1)}$ be the Fourier transform, restricted to the unit 3-sphere ${\rm Sp}(1)\subset \HH\cong\wh{\RR}^4$. It is a gap-continuous family of self-adjoint Fredholm operators, which is homotopically non-trivial, and its Fermi gerbe has Dixmier--Douady invariant a generator of $H^3({\rm Sp}(1),\ZZ)$.
\end{corollary}

\subsection{Topological Fermi surface}
Each $\wt{h}^{\rm W}(q;\varGamma)$ in the Fourier decomposition of $\wt{H}^{\rm W}(\varGamma)$ has essential spectral gap $(-|q|,|q|)$. From Lemma \ref{lem:gap.vanishing}, we see that the non-vanishing DD-invariant of the Fermi gerbe for $\wt{H}^{\rm W}(\varGamma)$ (Corollary \ref{cor:Weyl.DD.invariant}) implies that the common essential spectral gap $(-1,1)$ of $\{\wt{h}^{\rm W}(q;\varGamma)\}_{q\in{\rm Sp}(1)}$ must get completely filled up by the discrete spectra of the $\wt{h}^{\rm W}(q;\varGamma)$. Physically, one says that the ``bulk spectral gap is filled up by edge-states''. Actually, the Fermi gerbe could be defined over $\wh{\RR}^4{\setminus}\{0\}$, which retracts to the unit sphere ${\rm Sp}(1)$. So over \emph{any} 3-sphere in $\wh{\RR}^4$ which encloses the origin, the same essential spectral gap-filling must occur.

Let $\mu\in(-1,1)$ be a \emph{Fermi level}. We are interested in the subset 
\begin{equation*}
S_\mu:=\{q\in\HH\cong\wh{\RR}^4\,:\, \mu\in\sigma_{\rm d}(\wt{h}^{\rm W}(q;\varGamma))\},
\end{equation*}
called the \emph{Fermi surface} of edge states, at Fermi level $\mu$, for the half-space Weyl Hamiltonian $\wt{H}^{\rm W}(\varGamma)$. The qualitative structure of $S_\mu$ may be determined as follows. Whenever $|q|\leq|\mu|$, the operator $\wt{h}^{\rm W}(q;\varGamma)$ has $\mu$ in its essential spectrum. But at any momentum radius $\rho>|\mu|$, the non-trivial Fermi gerbe for $\{\wt{h}^{\rm W}(q;\varGamma)\}_{|q|=\rho}$ forces discrete spectra to fill up the essential spectral gap $(-\rho,\rho)$, and the value $\mu$ is certainly attained somewhere at this radius. So each radius $\rho>|\mu|$ sphere contributes a non-empty set to the Fermi surface $S_\mu$, and we deduce that $S_\mu$ connects the critical 3-sphere $\rho=|\mu|$ to infinity. In particular, for $\mu=0$, the Fermi surface $S_0$ connects the origin to infinity.

\vspace{1em}
Of course, we could have used Eq.\ \eqref{eqn:transformed.family} and the exact spectral computation, Prop.\ \ref{prop:discrete.spectrum}, to deduce directly that $S_\mu$ is explicitly given by ${\rm Re}(\overline{q}\varGamma)=\mu$ with $|q|>|\mu|$, which is a punctured hypersurface. However, the key point of discovering the non-trivial DD-invariant of the underlying Fermi gerbe is that we may exploit its homotopy invariance to deduce that similar qualitative features of the discrete spectra continue to hold, \emph{even for suitably perturbed families for which there is generally no hope of solving the spectral problems exactly}.

In more detail, consider a perturbed family $\{\wt{h}^{\rm W}(q;\varGamma)+V(q)\}_{q\in\HH}$, where
 $V(q)$ is some relatively compact perturbation of $\wt{h}^{\rm W}(q;\varGamma)$, thus preserving its essential spectrum. For instance, $V(q)=V$ could be an arbitrary continuous $4\times 4$ Hermitian matrix-valued potential function which vanishes as $z\rightarrow \infty$. The perturbed Fermi surface could be very complicated and difficult to compute exactly. Nevertheless, the Fermi gerbe is still well-defined for the perturbed family, and by turning off the perturbation to homotope the family to the unperturbed one, we see that its DD-invariant remains non-trivial. Lemma \ref{lem:gap.vanishing} still applies, and we conclude that the perturbed Fermi surface (at Fermi level $\mu=0$) still connects the origin to infinity.

\begin{remark}[Topology of Fermi surface is not preserved]
For the boundary condition $\Gamma={\rm diag}(i,i)\in{\rm U}(2)$ (which is not in ${\rm Sp}(1)$), it is possible to show that the Fermi surface $S_0$ for $\wt{H}^{\rm W}(\Gamma)$ is the positive $p_1$-axis (thus it looks like a ``traditional'' Fermi arc), whereas it was a punctured hypersurface when a boundary condition $\varGamma\in{\rm Sp}(1)$ was chosen. Thus, the topology of the Fermi surface as a bare topological space, in the more simplistic sense of homeomorphism, homotopy/homology classes etc., can depend sensitively on the choice of boundary conditions, Fermi level, and/or perturbations. We need a more sophisticated geometric object, such as the Fermi gerbe, to extract a topological invariant protecting the Fermi surface.
\end{remark}

\begin{remark}\label{rem:Weyl.semimetal}[$\mathsf{PT}$-symmetric Weyl semimetals \cite{MT}]
Since $Q(\vect{p})$ are the most general $4\times 4$ Hermitian matrices commuting with $\Theta$, four-band tight-binding Hamiltonians with $\mathsf{PT}$ symmetry are equivalently families $\TT^5\ni k\mapsto \vect{p}(k)\cdot\vect{\gamma}=Q(\vect{p}(k))$, where $\TT^5$ is the Brillouin torus in $5$ dimensions, and $k\mapsto \vect{p}(k)$ is some 5-component vector field over $\TT^5$. The spectrum of $Q(\vect{p}(k))$ is simply $\pm\sqrt{|\vect{p}(k)|}$, with each eigenvalue twofold degenerate due to the quaternionic structure (Kramers' pairing). The eigenvalues cross at the zeroes of the vector field $\vect{p}$. If the local index of $\vect{p}$ at such a zero $k^*$ is non-vanishing (without loss, set it to $+1$), the four-band crossing is topologically protected, and there must be a second zero $k^\star$ somewhere else, with the opposite index, due to the Poincar\'{e}--Hopf theorem. In that case, $k^*$ is called a (generalised) \emph{Weyl point} for a $\mathsf{PT}$-symmetric \emph{Weyl semimetal} Hamiltonian $k\mapsto Q(\vect{p}(k))$. The linear expansion of $Q(\vect{p}(k))$ near $k^*$ is, up to a change of coordinates, just $\vect{p}\cdot\vect{\gamma}$. In this sense, the Weyl Hamiltonian (whose Fourier transform is $\vect{p}\cdot\vect{\gamma}$) is often taken to be the continuum differential operator model for a tight-binding Weyl semimetal near a Weyl point. When the tight-binding model Hamiltonian is truncated to the half-lattice $\ZZ^4\times\NN$, there is likewise a projection $\pi:\TT^5\rightarrow\TT^4$ of Brillouin tori. It can be proved, using a generalised Toeplitz index theorem, that the projected Weyl points $\pi(k^*), \pi(k^\star)$ are connected in $\TT^4$ by a collection of edge-localised states \cite{MT,GCT}. Near $k^*$ (and similarly, near $k^\star$), the Fermi surface of $\wt{H}^{\rm W}(\varGamma)$ models the edge states of the half-space tight-binding Hamiltonian, as far as their topological protection is concerned. Here, we should note that the boundary condition $\varGamma$ in the continuum model is not easily translated to the effective tight-binding model, so it is of importance that the Fermi surface in the continuum model is topologically protected, regardless of $\varGamma$.
\end{remark}

\section*{Acknowledgments}
G.C.T. acknowledges support from Australian Research Council grant DP200100729, the University of Adelaide for hosting him as a visitor, and helpful discussions with M.\ Ludewig. A.L.C. thanks U. C. Emir. On behalf of all authors, the corresponding author states that there is no conflict of interest.

\appendix
\section{Quaternion conventions}\label{sec:appendix.quaternion}

The quaternion algebra $\HH$ is generated by three anticommuting square roots of $-1$, labelled $I, J, K$. It can be represented with $2\times 2$ complex matrices, e.g.
\begin{equation}
\HH\ni q=\underbrace{q_r}_{{\rm Re}(q)}+\underbrace{bI+cJ+dK}_{i\cdot{\rm Im}(q)} \longleftrightarrow \begin{pmatrix} q_r+ib & -c-id \\ c-id & q_r-ib \end{pmatrix}.\label{eqn:quaternion.complex.matrix}
\end{equation}
In terms of Pauli matrices, Eq.\ \eqref{eqn:Pauli.matrices}, the imaginary part of $q$ is ${\rm Im}(q)=\vect{q}\cdot\vect{\sigma}\equiv \sum_{j=1}^3 q_j\sigma_j$, where $\vect{q}=(-d,-c,b)$.

Quaternion conjugation $q\mapsto\overline{q}=q_r-i\vect{q}\cdot\vect{\sigma}$ corresponds to the Hermitian adjoint in this representation. The norm is given by $|q|^2=q\overline{q}=\overline{q}q=q_r^2+|\vect{q}|^2$, and the 3-sphere of unit quaternions is then identified with ${\rm SU}(2)$. On $\CC^2$, there is a standard quaternionic structure given by the operator $\Theta=-i\sigma_2\circ\kappa$ (where $\kappa$ denotes complex conjugation), which is antiunitary, squares to $-1$, and commutes with the left multiplication by $q$ in Eq.\ \eqref{eqn:quaternion.complex.matrix}. In physics, $\Theta$ may be interpreted as a fermionic time-reversal operator. A quaternionic basis vector for $\CC^2$ is a vector $\mathsf{e}$ such that $\{\mathsf{e},\Theta\mathsf{e}\}$ is an orthonormal basis (over $\CC$) for $\CC^2$. 
It allows us to view $\CC^2$ as a quaternionic vector space on which $i,\Theta$ generate the (right) quaternionic scalar multiplication, and also to identify ${\rm SU}(2)\cong {\rm Sp}(1)$. Similarly, a quaternionic structure on $\CC^{2n}$ is an antiunitary squaring to $-1$, and a quaternionic basis $\{\mathsf{e}_j\}_{j=1}^n$ gives an identification $\CC^{2n}\cong\HH^n$.
One should not confuse the role of $\HH\subset M_2(\CC)$ as an algebra of operators, and as a vector space $\HH\cong\CC^2$.

\section{Bundle gerbes}\label{appendix:bundle.gerbes}
Bundle gerbes are, in a sense, a generalisation of line bundles. We recall the local description of bundle gerbes, and refer to \cite{Murray, MS, Hitchin} for detailed treatments. Let $\mathcal{U}=\{U_i\}_{i\in I}$ be an open cover of $X$. The data of a bundle gerbe over $X$ comprise, for each pair $U_i, U_j\in\mathcal{U}$, a Hermitian line bundle $\mathcal{L}_{ij}$ over $U_i\cap U_j$. It is assumed that $\mathcal{L}_{ii}$ is trivial for each $i\in I$, and that on each triple overlap $U_i\cap U_j\cap U_k$, there is a unitary isomorphism $\phi_{ijk}:\mathcal{L}_{ij}\otimes\mathcal{L}_{jk}\rightarrow\mathcal{L}_{ik}$ such that ``associativity'' holds on quadruple overlaps $U_i\cap U_j\cap U_k\cap U_l$,
\begin{equation*}
\begin{CD}
\mathcal{L}_{ij} \otimes \mathcal{L}_{jk} \otimes \mathcal{L}_{kl}
@>{1 \otimes \phi_{jkl}}>> 
\mathcal{L}_{ij} \otimes \mathcal{L}_{jl} \\
@V{\phi_{ijk} \otimes 1}VV @VV{\phi_{ijl}}V \\
\mathcal{L}_{ik} \otimes \mathcal{L}_{kl} 
@>{\phi_{ikl}}>>
\mathcal{L}_{il}.
\end{CD}
\end{equation*}
In particular, $\mathcal{L}_{ij}\cong\mathcal{L}_{ji}^*$. The isomorphisms $\phi_{ijk}$ may be specified by a \v{C}ech 2-cocycle, and define (after passing to refinements) a cohomology class in $H^3(X,\ZZ)$ called the Dixmier--Douady invariant of the bundle gerbe. It is the analogue for bundle gerbes of the Chern class of a line bundle.


\begin{thebibliography}{99}
\bibliographystyle{plain}
\bibitem{AMV} Armitage, N.P., Mele, E.J., Vishwanath, A.: Weyl and Dirac semimetals in three-dimensional solids. Rev. Mod. Phys. {\bf 90}(1) 015001 (2018)
\bibitem{AS}
Atiyah, M.F., Singer, I.M.: Index theory for skew-adjoint {F}redholm operators. Inst. Hautes \'{E}tudes Sci. Publ. Math. {\bf 37} 5--26 (1969)
\bibitem{ASSS}
Avron, J.E., Sadun, L., Segert, J., Simon, B.: Chern numbers, quaternions, and Berry's phases in Fermi systems. Commun. Math. Phys. {\bf 124} 595--627 (1989)
\bibitem{BLP} 
Booss-Bavnbek, B., Lesch, M., Phillips, J.: Unbounded Fredholm operators and spectral flow. Canad. J. Math. {\bf 57}(2) 225--250 (2005)
\bibitem{Br}
Brylinski, J.-L.: Loop Spaces, Characteristic Classes and Geometric Quantization. Birkh\"{a}user, Boston-Basel-Berlin (1993)
\bibitem{CJM}
Carey, A.L., Johnson, S., Murray, M.K.: Holonomy on D-branes.
J. Geom. Phys. {\bf 52}(2) 186--216 (2004)
\bibitem{CMUniversal}
Carey, A.L., Mickelsson, J.: The universal gerbe, Dixmier--Douady class, and gauge theory. Lett. Math. Phys. {\bf 59} 47--60 (2002)
\bibitem{CMM}
Carey, A.L., Murray, M.K., Mickelsson, J.: Index theory, gerbes, and Hamiltonian quantization. Commun. Math. Phys. {\bf 183} 707--722 (1997)
\bibitem{CMFaddeev}
Carey, A.L., Murray, M.K.: Faddeev's anomaly and bundle gerbes. Lett. Math. Phys. {\bf 37} 29--36 (1996)
\bibitem{CMM} 
Carey, A.L., Mickelsson, J., Murray, M.K.: Bundle gerbes applied to quantum field theory. Rev. Math. Phys. {\bf 12}(01) 65--90 (2000)
\bibitem{Gawedzki}
Gaw\c{e}dzki, K.: Square root of gerbe holonomy and invariants of time-reversal-symmetric topological insulators. J. Geom. Phys. {\bf 120} 169--191 (2007)
\bibitem{GomiTauber}
Gomi, K., Tauber, C.: Eigenvalue crossings in Floquet topological systems. Lett. Math. Phys. {\bf 110} 465--500 (2020)
\bibitem{GT-gerbe}
Gomi, K., Thiang, G.C.: `Real' gerbes and Dirac cones of topological insulators. \href{https://www.arxiv.org/abs/2103.05350}{arXiv:2103.05350}
\bibitem{HWK}
Hashimoto, K., Wu, X., Kimura, T.: Edge states at an intersection of edges of a topological material. Phys. Rev. B {\bf 95} 165443 (2017)
\bibitem{Hitchin}
Hitchin, N.: Lectures on special Lagrangian submanifolds. In: Vafa, C., Yau, S.-T. (eds.) Winter School on Mirror Symmetry, Vector Bundles and Lagrangian Submanifolds, vol.\ {\bf 23} of AMS/IP Stud. Adv. Math., pp.\ 151--182. Amer. Math. Soc., Providence, RI (2001)
\bibitem{Joachim} 
Joachim, M.: Unbounded Fredholm operators and $K$-theory. In: Farrell, F.T., L\"{u}ck, W. (eds.) High-dimensional manifold topology, pp.\ 177--199. World Sci. Publishing (2003)
\bibitem{LM}
Lawson, B., Michelsohn, M.: Spin Geometry, Princeton Univ. Press, 1989.
\bibitem{MT} Mathai, V., Thiang, G.C.: Differential topology of semimetals. Commun. Math. Phys. {\bf 355} 561--602 (2017)
\bibitem{Murray}
Murray, M.K.: Bundle gerbes. J. London Math. Soc. {\bf 2}(54) 403--416 (1996)
\bibitem{MS} Murray, M.K., Stevenson D.: Bundle gerbes: stable Isomorphism and local theory. J. London Math. Soc. {\bf 62}(3) 925--937 (2000)
\bibitem{Ozawa}
Ozawa, T., Price, H.M.: Topological quantum matter in synthetic dimensions. Nature Rev. Phys. {\bf 1} 349--357 (2019) 
\bibitem{Palumbo}
Palumbo, G., Goldman, N.: Revealing tensor monopoles through quantum-metric measurements. Phys. Rev. Lett. {\bf 121} 170401 (2018)
\bibitem{Phillips} Phillips, J.: Self-adjoint Fredholm operators and spectral flow. Canad. Math. Bull. {\bf 39}(4) 460--467 (1996)
\bibitem{PS}
Pressley, A., Segal, G.: Loop Groups. Clarendon Press, Oxford, 1986.
\bibitem{RS1} 
Reed, M., Simon, B.: Methods of Mathematical Physics, vol I, Acad. Press (1980).
\bibitem{RS2} 
Reed, M., Simon, B.: Methods of Mathematical Physics, vol II, Acad. Press (1975).
\bibitem{Tan}
Tan, X. et al.: Experimental observation of tensor monopoles with a superconducting qudit. Phys. Rev. Lett. {\bf 126} 017702 (2021)
\bibitem{GCT}
Thiang, G.C.: On spectral flow and Fermi arcs. Commun. Math. Phys. (2021). \href{https://doi.org/10.1007/s00220-021-04007-z}{https://doi.org/10.1007/s00220-021-04007-z}
\bibitem{Viennot}
Viennot, D.: Geometric phases in adiabatic Floquet theory, Abelian gerbes and Cheon's anholonomy. J. Phys. A Math. Theor. {\bf 42} 395302 (2009)
\end{thebibliography}
\end{document}